\documentclass[pra,aps,showpacs]{revtex4}
\usepackage{amsmath,amsfonts,amssymb,caption,hyperref,color,epsfig,graphics,graphicx,latexsym,mathrsfs,revsymb,theorem,url,epstopdf}

\hypersetup{colorlinks,linkcolor={blue},citecolor={blue},urlcolor={red}}

\usepackage{hyperref}

\newtheorem{definition}{Definition}
\newtheorem{proposition}[definition]{Proposition}
\newtheorem{Lemma}[definition]{Lemma}

\newtheorem{Theorem}[definition]{Theorem}
\newtheorem{Corollary}[definition]{Corollary}
\newtheorem{conjecture}[definition]{Conjecture}

\newtheorem{remark}[definition]{Remark}
\newtheorem{example}[definition]{Example}
\newtheorem{question}[definition]{Question}

\def\squareforqed{\hbox{\rlap{$\sqcap$}$\sqcup$}}
\def\qed{\ifmmode\squareforqed\else{\unskip\nobreak\hfil
		\penalty50\hskip1em\null\nobreak\hfil\squareforqed
		\parfillskip=0pt\finalhyphendemerits=0\endgraf}\fi}
\def\endenv{\ifmmode\;\else{\unskip\nobreak\hfil
		\penalty50\hskip1em\null\nobreak\hfil\;
		\parfillskip=0pt\finalhyphendemerits=0\endgraf}\fi}
\newenvironment{proof}{\noindent \textbf{{Proof.~} }}{\qed}
\def\Dbar{\leavevmode\lower.6ex\hbox to 0pt
	{\hskip-.23ex\accent"16\hss}D}
\makeatletter
\def\url@leostyle{%
	\@ifundefined{selectfont}{\def\UrlFont{\sf}}{\def\UrlFont{\small\ttfamily}}}
\makeatother
\def\bcj{\begin{conjecture}}
	\def\ecj{\end{conjecture}}
\def\bcr{\begin{corollary}}
	\def\ecr{\end{corollary}}
\def\bd{\begin{definition}}
	\def\ed{\end{definition}}
\def\bea{\begin{eqnarray}}
\def\eea{\end{eqnarray}}
\def\bem{\begin{enumerate}}
	\def\eem{\end{enumerate}}
\def\bex{\begin{example}}
	\def\eex{\end{example}}
\def\bim{\begin{itemize}}
	\def\eim{\end{itemize}}
\def\bl{\begin{lemma}}
	\def\el{\end{lemma}}
\def\bma{\begin{bmatrix}}
	\def\ema{\end{bmatrix}}
\def\bpf{\begin{proof}}
	\def\epf{\end{proof}}
\def\bpp{\begin{proposition}}
	\def\epp{\end{proposition}}
\def\bqu{\begin{question}}
	\def\equ{\end{question}}
\def\br{\begin{remark}}
	\def\er{\end{remark}}
\def\bt{\begin{theorem}}
	\def\et{\end{theorem}}

\def\btb{\begin{tabular}}
	\def\etb{\end{tabular}}

\newcommand{\nc}{\newcommand}

\nc{\bbA}{\mathbb{A}} \nc{\bbB}{\mathbb{B}} \nc{\bbC}{\mathbb{C}}
\nc{\bbD}{\mathbb{D}} \nc{\bbE}{\mathbb{E}} \nc{\bbF}{\mathbb{F}}
\nc{\bbG}{\mathbb{G}} \nc{\bbH}{\mathbb{H}} \nc{\bbI}{\mathbb{I}}
\nc{\bbJ}{\mathbb{J}} \nc{\bbK}{\mathbb{K}} \nc{\bbL}{\mathbb{L}}
\nc{\bbM}{\mathbb{M}} \nc{\bbN}{\mathbb{N}} \nc{\bbO}{\mathbb{O}}
\nc{\bbP}{\mathbb{P}} \nc{\bbQ}{\mathbb{Q}} \nc{\bbR}{\mathbb{R}}
\nc{\bbS}{\mathbb{S}} \nc{\bbT}{\mathbb{T}} \nc{\bbU}{\mathbb{U}}
\nc{\bbV}{\mathbb{V}} \nc{\bbW}{\mathbb{W}} \nc{\bbX}{\mathbb{X}}
\nc{\bbZ}{\mathbb{Z}}

\nc{\bA}{{\bf A}} \nc{\bB}{{\bf B}} \nc{\bC}{{\bf C}}
\nc{\bD}{{\bf D}} \nc{\bE}{{\bf E}} \nc{\bF}{{\bf F}}
\nc{\bG}{{\bf G}} \nc{\bH}{{\bf H}} \nc{\bI}{{\bf I}}
\nc{\bJ}{{\bf J}} \nc{\bK}{{\bf K}} \nc{\bL}{{\bf L}}
\nc{\bM}{{\bf M}} \nc{\bN}{{\bf N}} \nc{\bO}{{\bf O}}
\nc{\bP}{{\bf P}} \nc{\bQ}{{\bf Q}} \nc{\bR}{{\bf R}}
\nc{\bS}{{\bf S}} \nc{\bT}{{\bf T}} \nc{\bU}{{\bf U}}
\nc{\bV}{{\bf V}} \nc{\bW}{{\bf W}} \nc{\bX}{{\bf X}}
\nc{\bZ}{{\bf Z}}

\nc{\cA}{{\cal A}} \nc{\cB}{{\cal B}} \nc{\cC}{{\cal C}}
\nc{\cD}{{\cal D}} \nc{\cE}{{\cal E}} \nc{\cF}{{\cal F}}
\nc{\cG}{{\cal G}} \nc{\cH}{{\cal H}} \nc{\cI}{{\cal I}}
\nc{\cJ}{{\cal J}} \nc{\cK}{{\cal K}} \nc{\cL}{{\cal L}}
\nc{\cM}{{\cal M}} \nc{\cN}{{\cal N}} \nc{\cO}{{\cal O}}
\nc{\cP}{{\cal P}} \nc{\cQ}{{\cal Q}} \nc{\cR}{{\cal R}}
\nc{\cS}{{\cal S}} \nc{\cT}{{\cal T}} \nc{\cU}{{\cal U}}
\nc{\cV}{{\cal V}} \nc{\cW}{{\cal W}} \nc{\cX}{{\cal X}}
\nc{\cZ}{{\cal Z}}

\nc{\hA}{{\hat{A}}} \nc{\hB}{{\hat{B}}} \nc{\hC}{{\hat{C}}}
\nc{\hD}{{\hat{D}}} \nc{\hE}{{\hat{E}}} \nc{\hF}{{\hat{F}}}
\nc{\hG}{{\hat{G}}} \nc{\hH}{{\hat{H}}} \nc{\hI}{{\hat{I}}}
\nc{\hJ}{{\hat{J}}} \nc{\hK}{{\hat{K}}} \nc{\hL}{{\hat{L}}}
\nc{\hM}{{\hat{M}}} \nc{\hN}{{\hat{N}}} \nc{\hO}{{\hat{O}}}
\nc{\hP}{{\hat{P}}} \nc{\hR}{{\hat{R}}} \nc{\hS}{{\hat{S}}}
\nc{\hT}{{\hat{T}}} \nc{\hU}{{\hat{U}}} \nc{\hV}{{\hat{V}}}
\nc{\hW}{{\hat{W}}} \nc{\hX}{{\hat{X}}} \nc{\hZ}{{\hat{Z}}}

\nc{\hn}{{\hat{n}}}

\def\max{\mathop{\rm max}}
\def\min{\mathop{\rm min}}

\def\tr{\mathop{\rm Tr}}

\newcommand{\bra}[1]{\langle#1|}
\newcommand{\ket}[1]{|#1\rangle}

\def\Dbar{\leavevmode\lower.6ex\hbox to 0pt
	{\hskip-.23ex\accent"16\hss}D}
\begin{document}
	\title{Monogamy relations for the generalized W class states beyond qubits}
	
	\author{Xian Shi}\email[]
	{shixian01@buaa.edu.cn}
	\affiliation{School of Mathematical Sciences, Beihang University, Beijing 100191, China}

	%\author{Yi Shen}\email[]
	%{yishen@buaa.edu.cn}
	%\affiliation{School of Mathematics and Systems Science, Beihang University, Beijing 100191, China}
	%
	%\author{Yize Sun}
	%\affiliation{School of Mathematics and Systems Science, Beihang University, Beijing 100191, China}
	
	%\author{Lijun Zhao}
	%\affiliation{School of Mathematics and Systems Science, Beihang University, Beijing 100191, China}
	
	%\author{Yumin Guo}
	%\affiliation{School of Mathematical Sciences, Capital Normal University, Beijing 100048, China}
	
	\date{\today}
	

\begin{abstract}
	\indent In this article, we consider the monogamy relations for the generalized W class states. Here we first present an analytical formula on Tsallis-$q$ entanglement and Tsallis-$q$ entanglement of assistance of a reduced density matrix of the generalized W class states, then we present a monogamy relation in terms of the  T$q$EE for the GW state, and we also present generalized monogamy relations in terms of T$q$EE for the GW states. At last, we present tighter polygamy relations in terms of T$q$EEoA when $q=2$ for the GWV states.
\end{abstract}

\maketitle
\section{introduction}
\indent Quantum entanglement \cite{horodecki2009quantum} is an essential feature of quantum information theory, which distinguishes the quantum from classical theory. One of the fundamental differences between entanglement and classical relations is that there eixsts some restrictions on its distribution and sharability \cite{coffman2000distributed,terhal2004entanglement}. This property is known as monogamy of entanglement (MoE).  Monogamy relations is valueable on the frustration effects observed in condensed matter physics \cite{ma2011quantum}. MoE is also a key ingredient to make quantum cryptography secure as it quantifies how much information an eavesdropper could potentially obtain about the secret key to be extracted \cite{masanes2009universally,renes2006generalized}.\\
\indent Mathematically, MoE can be represented as in terms of some entanglement measure $\mathcal{E}$ for a three-party system $\rho_{ABC}$
\begin{align}
\mathcal{E}_{A|BC}\ge\mathcal{E}_{AB}+\mathcal{E}_{AC}.
\end{align} 
This property was first shown by Coffman $et$ $al.$ \cite{coffman2000distributed} in terms of the squared concurrence for a three-qubit mixed state $\rho_{ABC},$ here we can denote this inequality as CKW inequality. It was generated for $n$-qubit systems in terms of the squared concurrence later \cite{osborne2006general}.  Then this relation is generalized in terms of the T$q$EE \cite{san2010tsallis, luo2016general}, the Renyi-$\alpha$ entropy \cite{kim2010monogamy}, and the unified entropy \cite{san2011unified} for multi-qubit systems. In 2014, Regula $et$ $al.$  conjectured the nonnegativity of n-tangle when $n\ge 3$, thus they proposed a stronger monogamy inequality   \cite{regula2014strong}. Recently, the generalized monogamy relations are presented for $n$-qubit systems \cite{zhu2015,san2016generalized,jin2018tighter,Zhu2018}.\\
\indent  However, the CKW inequality is invalid for higher dimensional systems in terms of the squared concurrence \cite{ou2007violation}. In 2016, Lancien $et$ $al.$ even showed there exists multi-partite higher dimensional systems donot satisfy any nontrivial monogamy relations in terms of a whole additive entanglement measures \cite{lancien2016should}. Up to date, it seems only one known entanglement measure, the squashed entanglement, is monogamous for arbitrary dimensional systems \cite{christandl2004squashed}.  And there are results on some set of states satisfying the monogamy relations in higher dimensional systems. In 2008, Kim and Sanders showed the generalized W class (GW) states satisfying the monogamy inequality in terms of the squared concurrence \cite{san2008generalized}. In 2015, Choi and Kim showed that the superposition of the generalized W-class states and the vacuum (GWV) states satisfy the strong monogamy inequality in terms of the squared convex roof extended negativity \cite{choi2015negativity}. In 2016, Kim showed that a partially coherent superposition (PCS) of a generalized W-class
state and the vacuum saturates the strong monogamy inequality  this result is interesting, as it is the first class of the mixed states in multipartite higher dimensional systems that satisfy the strong monogamy inequality in terms of the squared convex roof extended negativity \cite{san2016strong},.\\
\indent As a generalization of von Neumann entropy, Tsallis-q entropy plays an important role in quantum information theory. It can be used to provide criterion for separability of compound quantum systems \cite{rossignoli2002generalized,nayak2017biseparability}, and it was used to generalize the global quantum discord  \cite{chi2013generalized}, there provides a sufficient condition for an n-party quantum state to be monogamous. Furthermore, M. Wajs $et$ $al.$ showed that the entropic bell inequalities in terms of the classical Tsallis-q entropy can be used to investigate the nonlocal corrections which is more suitable than the Shannon entropy \cite{wajs2015information}. \\
\indent In this article, we consider the monogamy relations in terms of T$q$EE for the GW states. In section \uppercase\expandafter{\romannumeral2}, we present some preliminary knowledge on this article. In section \uppercase\expandafter{\romannumeral3}, we present our main results. In section \uppercase\expandafter{\romannumeral4}, we end with a summary.
\section{Preliminary Knowledge}
\indent Given a bipartite pure state $\ket{\psi}_{AB}=\sum_i\sqrt{\lambda_i}\ket{ii},$
the concurrence is defined as
\begin{align}
C(\ket{\psi}_{AB})=\sqrt{2(1-\tr\rho_A^2)}=\sqrt{2\sum_{i\ne j}\lambda_i\lambda_j},\label{p1}
\end{align}
where $\rho_A=\tr_B\rho_{AB}$. When $\rho_{AB}$ is a mixed state, its concurrence is defined as
\begin{align}
C(\rho_{AB})=\min_{\{p_i,\ket{\psi_i}\}}\sum_i p_iC(\ket{\psi_i}),
\end{align}
where the minimum takes over all the decompositions of $\rho_{AB}=\sum_i p_i\ket{\psi_i}_{AB}\bra{\psi_i}.$ As a dual quantity to the concurrence, we can define the concurrence of assistance (CoA) as 
\begin{align}
C_a(\rho_{AB})=\max_{\{p_i,\ket{\psi_i}\}}\sum_i p_iC(\ket{\psi_i}),
\end{align}
where the maximum takes over all the decompositions of $\rho_{AB}=\sum_i p_i\ket{\psi_i}_{AB}\bra{\psi_i}.$\\
\indent Assume $\ket{\psi}_{AB}$ is a pure state, its negativity is defined as
\begin{align}
N(\ket{\psi}_{AB})=||\ket{\psi}_{AB}\bra{\psi}||_1-1,\label{p2}
\end{align}
here $||\cdot||$ is the trace-norm \cite{vidal2002computable}. 
When we denote $\ket{\psi}_{AB}=\sum_i\sqrt{\lambda_i}\ket{ii},$  The equation (\ref{p2}) can be written as $$N(\ket{\psi}_{AB})=\sum_{i\ne j}\sqrt{\lambda_i\lambda_j} =[\tr\sqrt{\rho_A}]^2-1,$$ here $\rho_A=\tr_B\ket{\psi}_{AB}\bra{\psi}$. When $\rho_{AB}$ is a mixed state, its convex roof extended negativity is defined as
\begin{align}
\mathcal{N}(\rho_{AB})=\min_{\{p_i,\ket{\psi_i}\}}\sum_i p_iN(\ket{\psi_i}),
\end{align}
where the minimum takes over all the decompositions of $\rho_{AB}$ \cite{choi2015negativity}.\\
\indent For a pure state $\ket{\psi}_{AB}=\sum_i\sqrt{\lambda_i}\ket{ii},$ its Tsallis-$q$ entanglement (T$q$EE) is defined as
\begin{align}
T_q(\ket{\psi}_{AB})=\frac{1-\tr\rho_A^q}{q-1}=\frac{1-\sum_i \lambda_i^q}{q-1},\label{p3}
\end{align} 
for any $q>0,q\ne 1,$ here we denote that $\rho_A=\tr_B\ket{\psi}\bra{\psi}.$
Assume $\rho_{AB}$ is a mixed state, its T$q$EE is defined as \cite{san2010tsallis}
\begin{align}
T_q(\rho_{AB})=\min_{\{p_i,\ket{\psi_i}\}}\sum_i p_iT_q(\ket{\psi_i}),
\end{align}
where the minimum takes over all the decompositions of $\rho_{AB}.$ When $q\rightarrow 1,$ $T_q(\cdot)$ converges to the entanglement of formation $E(\cdot). $ As a dual concept of T$q$EE, its Tsallis-$q$ entanglement of assistance (T$q$EEoA) is defined as \cite{san2010tsallis}
\begin{align}
T_q^a(\rho_{AB})=\max_{\{p_i,\ket{\psi_i}\}}\sum_i p_iT_q(\ket{\psi_i}),
\end{align} where the minimum takes over all the decompositions of $\rho_{AB}.$\\
\indent From the equalities $(\ref{p1})$ and $(\ref{p3})$, we see that when $\ket{\psi}_{AB}=\sqrt{\lambda_0}\ket{00}+\sqrt{\lambda_1}\ket{11},$ $  C^2(\ket{\psi}_{AB})=4\lambda_0\lambda_1,
T_q(\ket{\psi}_{AB})=\frac{1-\lambda_0^q-\lambda_1^q}{q-1},$ from the above equalities, we have   
\begin{align}\label{Tc}
T_q(\ket{\psi}_{AB})=f_q(C^2(\ket{\psi}_{AB})),
\end{align}  
here
 \begin{align}\label{fq}
 f_q(x)=\frac{1}{q-1}[1-(\frac{1+\sqrt{1-x}}{2})^q-(\frac{1-\sqrt{1-x}}{2})^q].
 \end{align}
\par
Next we present some lemmas on the properties of the function $f_q$ in the equality $(\ref*{fq})$.
\begin{Lemma}\label{l1}\cite{luo2016general}
	The function $f_q^2(x)$ is a monotonously increasing and convex function when $q\in[\frac{5-\sqrt{13}}{2},\frac{5+\sqrt{13}}{2}].$  
\end{Lemma}
\begin{Lemma}\label{l2}\cite{luo2016general}
	The function $f_q(x)$ is a monotonously increasing and concave function when $q\in[\frac{5-\sqrt{13}}{2},2]\cup[3,\frac{5+\sqrt{13}}{2}].$
\end{Lemma}
Next we denote $g_q(y)=f_q(x^2)$.
\begin{Lemma}\label{l3}\cite{luo2016general}
	The function $g_q(y)$ is a monotonic increasing function of the variable $x$ for any $q>0$ and $0<x<1,$ it is a convex function of $x$ when $q\in[\frac{5-\sqrt{13}}{2},\frac{5+\sqrt{13}}{2}].$
\end{Lemma}
\par\indent Now let us recall the definition of the GW states $\ket{W_n^d}$,
\begin{align}
\ket{W_n^d}_{A_1\cdots A_n}=\sum_{i=1}^{d}(a_{1i}\ket{i0\cdots 0}+\cdots+a_{ni}\ket{00\cdots i}),\label{p5}
\end{align} 
where we assume $\sum_{i=1}^{d}\sum_{j=1}^{n}|a_{ji}|^2=1.$\\
\indent
Then we present a lemma proved by Choi and Kim \cite{choi2015negativity}.
\begin{Lemma} \label{CK} \cite{choi2015negativity}
	Let $\ket{\psi}_{A_1\cdots A_n}$ be the superposition of the generalized W class states and vacuum (GWV), that is,
	\begin{align}
	\ket{\psi}_{A_1\cdots A_n}=\sqrt{p}\ket{W_n^d}+\sqrt{1-p}\ket{00\cdots 0},
	\end{align}
	for $0\le p\le 1.$ Let $\rho_{A_{j_1}\cdots A_{j_m}} $ be a reduced density matrix of $\ket{\psi}_{A_1\cdots A_n}$ onto m-qudit subsystems $A_{j_1}\cdots A_{j_{m-1}}$ with $2\le m\le n-1.$ For any pure state decomposition of $\rho_{A_{j_1}\cdots A_{j_m}}$ such that 
	\begin{align}
	\rho_{A_{j_1}\cdots A_{j_m}}=\sum_k q_k\ket{\phi_k}_{A_{j_1}\cdots A_{j_m}}\bra{\phi_k}_{A_{j_1}\cdots A_{j_m}},
	\end{align} 
	$\ket{\phi_k}_{A_{j_1}\cdots A_{j_m}}$ is a GWV state.
\end{Lemma}

As each GWV state $\ket{\psi}_{A_{j_1}A_{j_2}\cdots A_{j_i}|A_{j_{i+1}}\cdots A_{j_m}}$ is a Schmidt rank 2 pure state by any partition and from the above lemma, we see that for any decomposition $\{p_i,\ket{\phi_i}_{A_{j_1}A_{j_2}\cdots A_{j_i}A_{j_{i+1}}\cdots A_{j_m}}\}$ of a reduced density matrix $\rho_{A_{j_1}\cdots A_{j_m}} $ of $\ket{\psi}_{A_1\cdots A_n},$ $\ket{\phi_i}_{A_{j_1}A_{j_2}\cdots A_{j_i}|A_{j_{i+1}}\cdots A_{j_m}}$ is a Schmidt rank 2 pure state.
\begin{Lemma}\label{san}\cite{san2008generalized}
	For any n-qudit generalized W-class state $\ket{\psi}_{AB_{1}\cdots B_{n-1},}$ in ($\ref{p5}$) and a partition $P=\{P_1,\cdots,P_m\}$ for the set of subsystems $S=\{A, B_{1}, \cdots, B_{n-1}\},$
	\begin{align}
	C^2_{P_1\cdots \overline{P_s}\cdots P_m}=\sum_{k\ne s}C^2_{P_sP_k}=\sum_{k\ne s}(C^a_{P_sP_k})^2,\label{c1}
	\end{align}
	and \begin{align}
	C_{P_sP_k}=C^a_{P_sP_k},\label{c2}
	\end{align}
	for all $k\ne s$.
\end{Lemma}\par
\indent Next we present a lemma used in the last part of the results we present.
\begin{Lemma}
	For real numbers $x\in[0,1]$  and $t\geq1$, we have 
	\begin{align}\label{lem}
	(1+t)^x\geq1+(2^x-1)t^x.
	\end{align}
\end{Lemma}
\begin{proof}
	The inequality (\ref{lem}) can be seen as a question to find the biggest value of the function on $t$, $$g_{x}(t)=\frac{(1+t)^x-1}{t^x}.$$
	Then we have when $t\ge 1,$ $$\frac{d g_{x}(t)}{d t}=xt^{-(x+1)}[1-(1+t)^{x-1}]\geq0,$$ hence $g_{x}(t)$ is an
	increasing function of $t$. Then when $x\in [0,1]$, $g_{x}(t)\geq g_{x}(1)$, i.e, $(1+t)^x\geq1+(2^x-1)t^x.$
\end{proof}
\section{Main Results}
\subsection{MONOGAMY RELATIONS IN TERMS OF TSALLIS ENTROPY FOR GENERALIZED W CLASS STSATES}
\indent First we present an analytical formula of the T$q$EE for the GW states, it is 
by the method in \cite{san2010tsallis}. 
\par \indent Assume $\rho_{AB}$ is a reduced density matrix of a pure GW state,  as $C(\rho_{AB})=C_a(\rho_{AB})$ \cite{san2008generalized}, and by the method in \cite{wootters1998entanglement}, we could find a decomposition $\{p_m,\ket{\theta_m}\}$ of $\rho_{AB}$ such that all of $C(\ket{\theta_m})$ are the same, that is, $C(\rho_{AB})=C_a(\rho_{AB})=\sum_m p_m C(\ket{\theta_m}).$ 
\begin{Theorem}\label{tq1}
	Assume $\rho_{A_{j_1}\cdots A_{j_m}}$ is a reduced density matrix of a pure GW state, then we have 
	\begin{align}
	T_q(\rho_{A_{j_1}|A_{j_2}\cdots A_{j_m}})=f_q(C^2(\rho_{A_{j_1}|A_{j_2}\cdots A_{j_m}})),
	\end{align}
	when $q\in[\frac{5-\sqrt{13}}{2},\frac{5+\sqrt{13}}{2}].$
\end{Theorem}
Similar to the method in \cite{luo2016general}, we could present an analytic formula for the T$q$EEoA for the GW states.
\begin{Theorem}\label{tq2}
	Assume $\rho_{A_{j_1}\cdots A_{j_m}}$ is a reduced density matrix of a pure GW state, then we have 
	\begin{align}
	T_q^a(\rho_{A_{j_1}|A_{j_2}\cdots A_{j_m}})=f_q(C^2(\rho_{A_{j_1}|A_{j_2}\cdots A_{j_m}})),
	\end{align}
	when $q\in[\frac{5-\sqrt{13}}{2},2]\cup[3,\frac{5+\sqrt{13}}{2}].$
\end{Theorem}
\begin{proof}
	\indent Here we denote $\rho_{A_{j_1}|A_{j_2}\cdots A_{j_m}}$ as $\rho_{AB}$ below. First we would prove $T^a_q(\rho_{AB})\le f_q(C^2(\rho_{AB})).$ Assume the decomposition $\{p_i,\ket{\psi_i}_{AB}\}$ is an optimal decomposition for the T$q$EEoA of $\rho_{AB},$ then we have
	\begin{align}\label{th11}
	T^a_q(\rho_{AB})=&\sum_ip_i T_q(\ket{\psi_i}_{AB})\nonumber\\
	=&\sum_ip_i f_q(C^2(\ket{\psi_i}_{AB}))\nonumber\\
	\le& f_q(\sum_ip_iC^2(\ket{\psi_i}_{AB}))\le f_q(C_a^2(\rho_{AB})),
	\end{align}
	where in the first equality, we use the definition of T$q$EEoA,  the first inequality holds due to the Lemma \ref{l2}, the second inequality is due to the \ref{l2} and the definition of $C_a^2(\rho_{AB})=\max\sum_k r_k C^2(\ket{\theta_k}),$ where the maximum takes over all the decompositions of $\{r_k,\theta_k\}.$\\
	\indent Then we will prove $T^a_q(\rho_{AB})\ge f_q(C^2(\rho_{AB})),$ we can obtain
	\begin{align}\label{th12}
	T^a_q(\rho_{AB})=&\sum_ip_i T_q(\ket{\psi_i}_{AB})\nonumber\\
	=&\sum_ip_i f_q(C(\ket{\psi_i}_{AB}))\nonumber\\
	\ge& f_q(\sum_i p_iC(\ket{\psi_i}_{AB}))\nonumber\\
	\ge& f_q(\sum_j s_jC(\ket{\zeta_j}_{AB}))=f_q(C(\rho_{AB})).
	\end{align}
	Here in the first inequality, we use the convexity of $f_q(C)$ when $q\in[\frac{5-\sqrt{13}}{2},\frac{5+\sqrt{13}}{2}]$ and the Cauchy-Schwartz inequality, in the second inequality, we denote that the decomposition $\{s_j,\ket{\zeta_j}\}$ is the optimal decomposition for the concurrence $C.$ Combing the inequalities $(\ref{th11})$ and $(\ref{th12}),$ we have
	\begin{align}
	f_q(C_a^2(\rho_{AB}))\ge T_q(\rho_{AB})\ge f_q(C(\rho_{AB})),\label{th11a}
	\end{align}
	then as $C(\rho_{AB})=C_a(\rho_{AB})$ \cite{san2008generalized}, the inequality $\ref{th11a}$ can be written as
	\begin{align}
		f_q(C_a^2(\rho_{AB}))\ge T_q(\rho_{AB})\ge f_q(C_a(\rho_{AB})),\label{th11b}
	\end{align}
	then by the Cauchy-Schwarz inequality, we have
	\begin{align}
	\sum_k r_kC^2(\ket{\theta_k})=	\sum_k (\sqrt{r_k})^2\sum_k (\sqrt{r_k}C(\ket{\theta_k})^2\le [ \sum_k r_kC(\theta_k)]^2
	\end{align}
	 We finish the proof.\\
\end{proof}
\par \indent Combing the Theorem \ref{tq1} and \ref{tq2}, we have the following results.
\begin{Theorem}\label{tq3}
Assume $\rho_{A_{j_1}\cdots A_{j_m}}$ is a reduced density matrix of a pure GW state, then we have 
\begin{align}
T_q^a(\rho_{A_{j_1}|A_{j_2}\cdots A_{j_m}})=T_q(\rho_{A_{j_1}|A_{j_2}\cdots A_{j_m}})=f_q(C^2(\rho_{A_{j_1}|A_{j_2}\cdots A_{j_m}})),
\end{align}
when $q\in[\frac{5-\sqrt{13}}{2},2]\cup[3,\frac{5+\sqrt{13}}{2}].$
\end{Theorem}
\indent Next we will provide a monogamy relation in terms of the T$q$EE for the reduced density matrix of the GW states when $q\in[\frac{5-\sqrt{13}}{2},\frac{5+\sqrt{13}}{2}].$ 
\begin{Theorem}\label{T2}
	Assume $\rho_{A_{j_1}A_{j_2}\cdots A_{j_m}}$ is the reduced density matrix of a GW state $\ket{\psi}_{A_1\cdots A_n},$ and here we denote $\{P_1,P_2,\cdots,P_k\}$ is a partition of the set $\{A_{j_1},A_{j_2},\cdots,A_{j_m}\},$ when $q\in[\frac{5-\sqrt{13}}{2},\frac{5+\sqrt{13}}{2}],$ we have the following monogamy inequality,
	\begin{align}
	T_q^2(\rho_{P_1|P_2\cdots P_k})\ge \sum_{i=2}^{k} T_q^2(\rho_{P_1P_i}).
	\end{align}
\end{Theorem}
\begin{proof}
	When $q\in[\frac{5-\sqrt{13}}{2},\frac{5+\sqrt{13}}{2}],$ we have
	\begin{align}
	T_q^2(\rho_{P_1|P_2\cdots P_k})=&f_q^2(C^2(\rho_{P_1|P_2\cdots P_k}))\nonumber\\
	=&f_q^2(\sum_{i=2}^{k} C^2(\rho_{P_1P_i}))\nonumber\\
	\ge& \sum_{i=2}^{k} f_q^2(C^2(\rho_{P_1P_i}))\nonumber\\
	=& \sum_{i=2}^{k} T_q^2(\rho_{P_1P_i}).
	\end{align}
	Here the second equality is due to the result $\sum_{i=2}^{k} C^2(\rho_{P_1P_i})=C^2(\rho_{P_1|P_2\cdots P_k})$, the first inequality is due to the Lemma \ref{l1} .\\
\end{proof}
\par \indent We can generalize the above monogamy relations in terms of the squared T$q$EE to the $\alpha$-th power of T$q$EE for the GW states when $\alpha\ge 2$.
\begin{Corollary}
	Assume $\rho_{A_{j_1}A_{j_2}\cdots A_{j_m}}$ is the reduced density matrix of a GW state $\ket{\psi_{A_1\cdots A_n}},$ and here we denote $\{P_1,P_2,\cdots,P_k\}$ is a partition of the set $\{A_{j_1},A_{j_2},\cdots,A_{j_m}\},$ when $q\in[\frac{5-\sqrt{13}}{2},\frac{5+\sqrt{13}}{2}],$ we have the following monogamy inequality,
	\begin{align}
	T_q^{\alpha}(\rho_{P_1|P_2\cdots P_k})\ge \sum_{i=2}^{k} T_q^{\alpha}(\rho_{P_1P_i}).
	\end{align}
	when $\alpha\ge 2.$\\
\end{Corollary}\par
\indent Similarly, we have the polygamy relation in terms of T$q$EE for the GW states.
\begin{Theorem}\label{T1}
	Assume $\rho_{A_{j_1}A_{j_2}\cdots A_{j_m}}$ is the reduced density matrix of a GW state $\ket{\psi_{A_1\cdots A_n}},$ and here we denote $\{P_1,P_2,\cdots,P_k\}$ is a partition of the set $\{A_{j_1},A_{j_2},\cdots,A_{j_m}\},$ when $q\in[\frac{5-\sqrt{13}}{2},2]\cup[3,\frac{5+\sqrt{13}}{2}],$ we have the following monogamy inequality,
	\begin{align}
	T_q(\rho_{P_1|P_2\cdots P_k})\le \sum_{i=2}^{k} T_q(\rho_{P_1P_i}).
	\end{align}
\end{Theorem}
\begin{proof}
	When $q\in[\frac{5-\sqrt{13}}{2},2]\cup[3,\frac{5+\sqrt{13}}{2}],$ we have
	\begin{align}
	T_q(\rho_{P_1|P_2\cdots P_k})=&f_q(C^2(\rho_{P_1|P_2\cdots P_k}))\nonumber\\
	=&f_q(\sum_{i=2}^{k} C^2(\rho_{P_1P_i}))\nonumber\\
	\le& \sum_{i=2}^{k} f_q(C^2(\rho_{P_1P_i}))\nonumber\\
	=& \sum_{i=2}^{k} T_q(\rho_{P_1P_i}),
	\end{align}
	where the first inequality is due to the concavity of $f_q(C^2)$ as a function of $C^2.$\\
\end{proof}
As when $\beta\le 1,x\ge 1$, $(1+x)^{\beta}\le 1+x^{\beta}$, we have
\begin{Corollary}
	Assume $\rho_{A_{j_1}A_{j_2}\cdots A_{j_m}}$ is the reduced density matrix of a GW state $\ket{\psi_{A_1\cdots A_n}},$ and here we denote $\{P_1,P_2,\cdots,P_k\}$ is a partition of the set $\{A_{j_1},A_{j_2},\cdots,A_{j_m}\},$ when $q\in[\frac{5-\sqrt{13}}{2},2]\cup[3,\frac{5+\sqrt{13}}{2}],$ we have the following monogamy inequality,
	\begin{align}
	T_q^{\beta}(\rho_{P_1|P_2\cdots P_k})\le \sum_{i=2}^{k} T_q^{\beta}(\rho_{P_1P_i}).
	\end{align}
	when $\beta\le 1.$\\
\end{Corollary}
\par \indent Next we present an example to show the meaning of the Theorem \ref{T2} and Theorem \ref{T1}.\\
\begin{example}
\begin{align}
&\ket{\psi}_{A_1A_2A_3A_4}\nonumber\\
=&0.3\ket{0001}+0.4\ket{0010}+{0.5}\ket{0100}+\sqrt{0.5}\ket{1000},
\end{align}
then 
\begin{align*}
\rho_{A_1A_2A_3}=&0.09\ket{000}\bra{000}+\ket{\phi}\bra{\phi}\nonumber\\
\ket{\phi}=&0.4\ket{001}+0.5\ket{010}+\sqrt{0.5}\ket{100},
\end{align*}
through computation, we have
\begin{align*}
T_q(\rho_{A_1A_2})=&f_q(\frac{\sqrt{2}}{2}),\nonumber\\
T_q(\rho_{A_1A_3})=&f_q(\frac{40-8\sqrt{17}}{125}),\nonumber
\end{align*}
\par \indent Then combing the Theorem \ref{T2} and Theorem \ref{T1}, we have 
\begin{align}
\sqrt{T^2_q(\rho_{A_1A_2})+T^2_q(\rho_{A_1A_3})}\le T_q(\rho_{A_1|A_2A_3})\le T_q(\rho_{A_1A_2})+T_q(\rho_{A_1A_3})\label{e1}
\end{align}
\begin{figure}
	\centering
	\includegraphics[width=110mm]{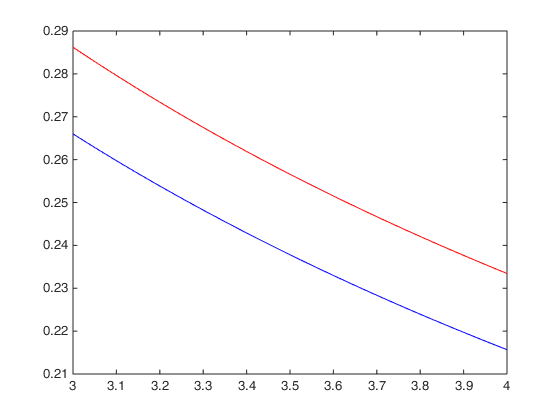}\\
	\caption{In this figure, the red line represents the function $T_q(\rho_{A_1A_2})+T_q(\rho_{A_1A_3})$ on variable $q$ when $q\in [3,4],$ the blue line represents the function $\sqrt{T^2_q(\rho_{A_1A_2})+T^2_q(\rho_{A_1A_3})}$ on variable $q$ when $q\in [3,4].$ }\label{}
\end{figure}
	\end{example}
\par
\indent From the proof of the monogamy inequalities for the GW states above, we see that the method can be generalized to derive monogamy inequalities for the GW states in terms of other entanglement measures, such as the squared convex roof extended negativity \cite{choi2015negativity},
the $R\acute{e}nyi-\alpha$ entropy for $\alpha$ in some region \cite{song2016general} and the unified entropy \cite{hu2006generalized,san2011unified}.\\
\indent Then we present a class of mixed states proposed by \cite{san2008generalized}, which is defined as
\begin{align}
\rho_{p}=p\ket{W_n^d}\bra{W_n^d}+(1-p)\ket{0}^{\otimes n}\bra{0},
\end{align}
where $p\in [0,1],$ this class of states were proved to satisfy the monogamy relation in terms of concurrence by \cite{san2008generalized}, and we show they also satisfy the Theorem \ref{T2} and \ref{T1}, first we consider the purification of $\rho_{p}$ such that
\begin{align}
\ket{\psi}_{p}=\sum_{i=1}^{n}\sqrt{p}(a_{1i}\ket{i00\cdots 0}+a_{2i}\ket{0i0\cdots 0}+\cdots+a_{ni}\ket{000\cdots i})\ket{0}+\sqrt{1-p}\ket{00\cdots0}\ket{\phi},
\end{align}
here we denote that $\ket{\phi}=\sum_{i=1}^d a_{n+1i}\ket{i}$ with $\sum_i |a_{n+1i}|^2=1,$ it is a GW state. Then we know that the above properties shown for the GW states are also valid for these mixed states.
\subsection{GENERALIZED MONOGAMY RELATIONS FOR THE GENERALIZED W CLASS STATES}
 \indent Here we recall the subadditivity in terms of Tsallis-q entropy.
 \begin{Lemma}\label{l4}\cite{audenaert2007subadditivity}
 	 For a bipartite state $\rho_{AB}$ on a Hilbert space $\mathcal{H}_A\otimes \mathcal{H}_B,$ then  when $q>1,$
 	 \begin{align}\label{t1}
 	\frac{1-\tr\rho_{AB}^q}{q-1}\le T_q(\rho_A)+T_q(\rho_B).
 	 \end{align}
 \end{Lemma}
\par \indent By the Lemma \ref{l4}, for a tripartite pure state $\ket{\psi}_{ABC}, $ we have 
\begin{align*}
\frac{1-\tr\rho_{AC}^q}{q-1}=&T_q(\rho_B)\nonumber\\
\ge& \frac{1-\tr\rho_{BC}^q}{q-1}-T_q(\rho_{C})\nonumber\\
=& T_q(\rho_A)-T_q(\rho_{C}),
\end{align*}
then combing the Lemma \ref{l4}, we have
\begin{align}
|T_q(\rho_A)-T_q(\rho_B)|\le \frac{1-\tr\rho_{AB}^q}{q-1}\le T_q(\rho_A)+T_q(\rho_B).\label{stq}
\end{align}
 \begin{Theorem}\label{T3}
 	Assume $\ket{\psi}_{PQR_1R_2\cdots R_{k-2}}$ is a GW state, when $q=2$ or $q=3,$ we have the following monogamy inequality,
 	\begin{align}
 	T_q(\ket{\psi}_{PQ|R_1R_2\cdots R_{k-2}})\le \sum_{i=1}^{k-2}[T_q(\rho_{PR_i})-T_q(\rho_{QR_i})].
 	\end{align}
 \end{Theorem}
\begin{proof}
	\begin{align}
	T_q(\ket{\psi}_{PQ|R_1R_2\cdots R_{k-2}})=&\frac{1-\tr\rho^q_{PQ}}{q-1}\nonumber\\
	\ge& T_q(\rho_P)-T_q(\rho_Q) \label{stq1}
	\end{align}
	In \cite{san2010tsallis}, the authors showed that when $q=2$ or $q=3,$ $x,y\in [0,1],$ 
	\begin{align}
	f_q(x^2+y^2)=f_q(x^2)+f_q(y^2). \label{san1}
	\end{align}
	 Then by the Lemma \ref{san} and Theorem \ref{tq1}, the (\ref{stq1}) becomes 
	\begin{align}
		T_q(\ket{\psi}_{PQ|R_1R_2\cdots R_{k-2}})=&\frac{1-\tr\rho^q_{PQ}}{q-1}\nonumber\\
	\ge& T_q(\rho_P)-T_q(\rho_Q) \nonumber\\
	=& 	T_q(\ket{\psi}_{P|QR_1R_2\cdots R_{k-2}})-	T_q(\ket{\psi}_{Q|PR_1R_2\cdots R_{k-2}})\nonumber\\
	=&T_q(\rho_{PQ})+\sum_iT_q(\rho_{PR_i})-T_q(\rho_{QP})-\sum_i T_q(\rho_{QR_i})\nonumber\\
	=&\sum_i[T_q(\rho_{PR_i})-T_q(\rho_{QR_i})].
	\end{align}
	Here the first inequality is due to the inequality (\ref{stq}), the second equality is due to the Lemma \ref{san} Theorem \ref{tq1} and (\ref{san1}).\\
\end{proof}\par
\indent Due to the results in Theorem \ref{tq3} and the results in \cite{san2016generalized}, we have the following corollary.
\begin{Corollary}
	Assume $\rho_{A_{j_1}A_{j_2}\cdots A_{j_m}}$ is the reduced density matrix of a GW state $\ket{\psi_{A_1\cdots A_n}},$ and here we denote $\{P_1,P_2,P_3\}$ is a partition of the set $\{A_{j_1},A_{j_2},\cdots,A_{j_m}\},$ when $q\in[1,2]\cup[3,\frac{5+\sqrt{13}}{2}],$ we have the following monogamy inequality,
	\begin{align}
	T_q(\rho_{P_1|P_2P_3})\le T_q(\rho_{P_2|P_1P_3})+T_q(\rho_{P_3|P_1P_2})
	\end{align}\\
\end{Corollary}\par
 \indent We also can generalize the above result to the following corollary.
\begin{Corollary}
Assume $\rho_{A_{j_1}A_{j_2}\cdots A_{j_m}}$ is the reduced density matrix of a GW state $\ket{\psi_{A_1\cdots A_n}},$ and here we denote $\{P_1,P_2,Q_1,Q_2,\cdots,Q_k\}$ is a partition of the set $\{A_{j_1},A_{j_2},\cdots,A_{j_m}\},$ when $q\in[1,2]\cup[3,\frac{5+\sqrt{13}}{2}],$ we have the following monogamy inequality,
\begin{align}
T_q(\rho_{P_1P_2|Q_1\cdots Q_k})\le 2T_q(\rho_{P_2|P_1})+\sum_{i=1}^kT_q(\rho_{P_1|Q_i})+\sum_{i=1}^kT_q(\rho_{P_2|Q_i}).
\end{align}
\end{Corollary}
\subsection{POLYGAMY RELATIONS IN TERMS OF TQEEOA FOR THE SUPERPOSITION OF THE GENERALISED W CLASS STATES AND VACUUM}
	\indent In this subsection, we present polygamy relations for the GWV states in terms of T$q$EEoA when q=2. Andin this subsection, we always note that $q=2$.  Assume $\ket{\psi}_{ABC}$ is a GWV state, then
by the results in \cite{shi2018monogamy}
\begin{align}
T_2(\ket{\psi}_{A|BC})\le T_2(\rho_{A|B})+T_2(\rho_{AC}).\label{ftqa}
\end{align}
Assume that $\rho_{ABC}$ is a reduced density matrix of a GWV state $\ket{\psi}_{ABCD},$ then we have 
\begin{align}
& T^a_q(\rho_{A|BC})\nonumber\\
=&\max \sum_i p_i T_q(\ket{\psi_i}_{A|BC})\nonumber\\
\le &\sum_i[ p_i T_q(\rho^i_{A|B})+p_iT_q(\rho^i_{A|C})]\nonumber\\
\le &T_q(\rho_{A|B})+T_q(\rho_{A|C}),
\end{align}
where in the first equality, we use the defintion of T$q$EEoA, in the first inequality, we use the inequality (\ref{ftqa}), the second inequality is due to the linearity of the operation of the partial trace,
\begin{align}
\sum_i p_i\rho_{AB}^i=\rho_{AB},
\end{align}
and we also use the definition of the T$q$EEoA,
\begin{align}
\sum_i p_iT_q(\rho_{A|B}^i)\le T_q^a(\rho_{A|B}),\hspace{2mm} \sum_i p_iT_q(\rho_{A|C}^i)\le T_q^a(\rho_{A|C}).
\end{align}
Then we have the following theorem.
\begin{Theorem}\label{tqa}
	Assume $\rho_{A_{j_1}\cdots A_{j_{m}}}$ is the reduced density matrix of a GWV state $\ket{\psi_{A_1A_2\cdots A_n}}$, here we denote that $\{P_1,P_2,\cdots,P_k\}$ is a partition of the set $\{j_1,j_2,\cdots,j_m\}$. When $q=2,$ we have the following polygamy inequality: 
	\begin{align}
	T_q^a(\rho_{P_1|P_2\cdots P_k})\le \sum_{i=2}^k T_q^a(\rho_{P_1|P_i}).
	\end{align}
\end{Theorem}
\indent Recently, results on the tighter monogamy inequalities in terms of concurrence \cite{jin2018tighter}, negativity \cite{jin2018tighter} for n-qubit systems and the entanglement of assistance for arbitrary dimensional systems \cite{San2018} are proposed. However, the results on the study of high dimensional systems are less, next we present a tighter polygamy inequality for the GW states.
\begin{Lemma}\label{lqa1}
	Let $\beta\in [0,1],x\in [0,1],$ then we have 
	\begin{align}
	(1+x)^{\beta}\le 1+(2^{\beta}-1)x^{\beta}
	\end{align}
\end{Lemma}
\begin{proof}
	let $t=\frac{1}{x},$ then the lemma is equivalent to get the maximum of $f(t)$ when t$\in[1,\infty),$
	\begin{align}
	f(t)=(1+t)^{\beta}-t^{\beta}.
	\end{align}
	As $t\in [1,\infty),$ and $f^{'}(t)\le 0,$ that is, when $t=1,$ $f(t)$ get the maximum $2^t-1.$ At last, When we replace t with $\frac{1}{x},$ we finish the proof. 
\end{proof}
\indent We know that any number $j\in \mathbb{N}^{+}$ can be written as
\begin{align}
j=\sum_{i=0}^{n-1}j_i 2^i,\label{tqa1}
\end{align}
here we assume $\log_2\le n,j_i\in {0,1}.$ According to the euqality (\ref{tqa1}), we have the following bijection:
\begin{align}
j\rightarrow &\vec{j}\nonumber\\
j\rightarrow &(j_0,j_1,\cdots,j_{n-1}),\nonumber
\end{align} 
then we denote its Hamming distance $w_H(\vec{j})$ as the number 1 of the set $\{j_0,j_1,\cdots,j_{n-1}\}$. Next we present the tighter polygamy inequality of the GW states in terms of T$q$EEoA.
\begin{Theorem}\label{tqa5}
	Let $\beta\in [0,1],$ when $q=2,$ assume $\rho_{PP_{j_0}\cdots P_{j_{m-1}}}$ is a reduced density matrix of a GWV state $\ket{\psi_{AA_1\cdots A_{n-1}}}$, then there exists an appropriate order of $P_{j_0}, P_{j_1}$ $\cdots$ $P_{j_{m-1}}$ such that
	\begin{align}
	[T_q^a(\rho_{P|P_{j_0}\cdots P_{j_{m-1}}})]^{\beta}\le \sum_{i=0}^{m-1}(2^{\beta}-1)^{w_H(\vec{j_i})}[T_q^a(\rho_{P|P_{j_i}})]^{\beta}
	\end{align}
\end{Theorem}
\begin{proof}
	In the process of the proof, we always order the partite $P_{j_0}, P_{j_1}$ $\cdots$ $P_{j_{m-1}}$ such that 
	\begin{align}
	T_q^a(\rho_{P|P_{j_i}})\ge T_q^a(\rho_{P|P_{j_{i+1}}}),i=0,1,\cdots,m-1.\label{tqa2}
	\end{align}
	Here we denote that the set $$A=\{\rho_{PP_{j_0}\cdots P_{j_{m-1}}}|\rho_{PP_{j_0}\cdots P_{j_{m-1}}} \textit{is a reduced density matrix of a GW state}\},$$ $$B=\{\rho_{PP_{j_0}\cdots P_{j_{m-1}}}=\gamma_{PP_{j_0}\cdots P_{j_{k-1}}}\otimes\ket{0_{m-k}}\bra{0_{m-k}}|\gamma_{PP_{j_0}\cdots P_{j_{m-1}}} \textit{is a reduced density matrix of a GW state}\}$$
	\indent Then we will prove the elements in the set $A\cup B$ is valid for the inequality (\ref{tqa2}).\\
	\indent 	Due to the Theorem \ref{tqa} and the definition of the set B, it is enough to prove 
	\begin{align}
	[\sum_{i=0}^{m-1}T_q^a(\rho_{P|P_{j_i}})]^{\beta}\le\sum_{i=0}^{m-1}(2^{\beta}-1)^{w_H(\vec{j_i})}[T_q^a(\rho_{P|P_{j_i}})]^{\beta}\nonumber
	\end{align}
	\indent	First we prove the theorem is correct when a tripartite mixed state $\rho_{ABC}$ is a reduced density matrix of a GWV state $\ket{\psi_{AA_1\cdots A_{n-1}}}$,
	\begin{align}
	& (T_q^a(\rho_{P|P_0P_1}))^{\beta}\nonumber\\
	\le &(T_q^a(\rho_{P|P_0}+T_q(\rho_{P|P_1}))^{\beta}\nonumber\\
	= &[T_q^a(\rho_{P|P_0})]^{\beta}\left[1+[\frac{T_q^a(\rho_{P|P_1})}{T_q^a(\rho_{P|P_0})}]\right]^{\beta}\nonumber\\
	\le &(T_q^a(\rho_{P|P_0}))^{\beta}+(2^{\beta}-1)(T_q^{a}(\rho_{P|P_1}))^{\beta},\label{tqa3}
	\end{align}
	here the first inequality is due to the Theorem \ref{tqa}, and when $a>c>0,b>0,$ $a^b>c^b,$ and the second inequality is due to the Lemma \ref{lqa1}.\\
	\indent Then we use the mathematical induction. First let us assume when $m<2^{n},$ the theorem is correct. Then we have when $m=2^n$, from the inequality (\ref{tqa3}), we have   
	\begin{align}
	&[\sum_{i=0}^{m-1}T_q^a(\rho_{P|P_{j_i}})]^{\beta}\nonumber\\
	=&\sum_{i=0}^{m/2-1}T_q^a(\rho_{P|P_{j_i}})^{\beta}[1+\frac{[\sum_{i=\frac{m}{2}}^{m-1}T_q^a(\rho_{P|P_{j_i}})]^{\beta}}{[\sum_{i=0}^{\frac{m}{2}-1}T_q^a(\rho_{P|P_{j_i}})]^{\beta}}]\nonumber\\
	\le &\sum_{i=0}^{\frac{m}{2}-1}[T_q^a(\rho_{P|P_{j_i}})]^{\beta}+(2^{\beta}-1)\sum_{i=\frac{m}{2}}^{m-1}[T_q^a(\rho_{P|P_{j_i}})]^{\beta}\nonumber\\
	\le& \sum_{i=0}^{\frac{m}{2}-1}(2^{\beta}-1)^{w_H(\vec{j_i})}(T_q^a(\rho_{P|P_{j_i}}))^{\beta}+\sum_{i=\frac{m}{2}}^{m-1}(2^{\beta}-1)*(2^{\beta}-1)^{w_H(\vec{j_i})-1}(T_q^a(\rho_{P|P_{j_i}}))^{\beta}\nonumber\\
	\le& \sum_{i=0}^{m-1}(2^{\beta}-1)^{w_H(\vec{j_i})}[T_q^a(\rho_{P|P_{j_i}})]^{\beta} \label{tqa4}
	\end{align}
	\indent When $m$ is an arbitrary number, we always can choose an $n\in \mathbb{N}^{+}$ such that $2^{n-1}\le m\le 2^n.$ Then we choose a $2^n+1$ party quantum state in the set B,
	\begin{align}
	\gamma_{PP_{j_0}\cdots P_{j_{2^n-1}}}=\rho_{PP_{j_0}\cdots P_{j_{m-1}}}\otimes \ket{0}_{2^n-m}\bra{0}.
	\end{align}
	\indent	Then due to the inequality (\ref{tqa4}), we have 
	\begin{align}
	\sum_{i=0}^{m-1}[T_q^a(\gamma_{P|P_{j_i}})]^{\beta}\le \sum_{i=0}^{m-1}(2^{\beta}-1)^{w_H(\vec{j_i})}[T_q^a(\gamma_{P|P_{j_i}})]^{\beta}
	\end{align}
	\indent From the definition of the state $\gamma_{PP_{j_0}\cdots P_{j_{2^n-1}}}$, we have 
	\begin{align}
	T_q^a(\gamma_{P|P_{j_0}\cdots P_{j_{2^n-1}}})=& T_q^a(\rho_{P|P_{j_0}\cdots P_{j_{m-1}}})\\
	T_q^a(\gamma_{P|P_{j_i}})=& T_q^a(\rho_{P|P_{j_i}}),i=0,1,\cdots,m-1,\\
	T_q^a(\gamma_{P|P_{j_i}})=& 0,i=m,m+1,\cdots,2^n-1,
	\end{align}
	\indent then we have
	\begin{align}
	&[T_q^a(\gamma_{P|P_{j_0}\cdots P_{j_{2^n-1}}})]^{\beta}\nonumber\\
	=& [T_q^a(\rho_{P|P_{j_0}\cdots P_{j_{m-1}}})]^{\beta}\nonumber\\
	\le& \sum_{i=0}^{2^n-1}(2^{\beta}-1)^{w_H(\vec{j_i})}[T_q^a(\gamma_{P|P_{j_i}})]^{\beta}\nonumber\\
	=&\sum_{i=0}^{m-1}(2^{\beta}-1)^{w_H(\vec{j_i})}[T_q^a(\rho_{P|P_{j_i}})]^{\beta}
	\end{align}
\end{proof}
\indent In the proof of the Theroem \ref{tqa5}, as we assume that $\beta\in [0,1],$ then 
\begin{Corollary}
	Let $\beta\in[0,1] $ and $\rho_{PP_{j_0}\cdots P_{j_{m-1}}}$ is a reduced density matrix of a GW state $\ket{\psi_{PP_1\cdots P_{n-1}}}$, when $q=2,$ then we have
	\begin{align}
	[T_q^a(\rho_{P|P_{j_0}\cdots P_{j_{m-1}}})]^{\beta}\le \sum_{i=0}^{m-1}[T_q^a(\rho_{P|P_{j_i}})]^{\beta}
	\end{align} 
\end{Corollary}
\indent At last, we present a tighter polygamy relation in terms of T$q$EEoA for GW states under some conditions we present.
\begin{Theorem}
	When $q=2,$	let $\beta\in[0,1]$ and $\rho_{PP_{j_0}\cdots P_{j_{m-1}}}$ is a reduced density matrix of a GWV state $\ket{\psi_{AA_1\cdots A_{n-1}}}$ , then When 
	\begin{align}
	T_q^a(\rho_{P|P_{j_i}})\ge \sum_{k=i+1}^{m-1} T_q^a(\rho_{P|P_{j_k}}),\label{ftqa0}
	\end{align} 
	we have
	\begin{align}
	[T_q^a(\rho_{P|P_{j_0}\cdots P_{j_{m-1}}})]^{\beta}\le \sum_{i=0}^{m-1}(2^{\beta}-1)^{i}[T_q^a(\rho_{P|P_{j_i}})]^{\beta}.
	\end{align}
\end{Theorem}
\begin{proof}
	\indent According to the Lemma \ref{lqa1}, we need to prove
	\begin{align}
	\sum_{i=0}^{m-1}[T_q^a(\rho_{P|P_{j_i}})]^{\beta}\le \sum_{i=0}^{m-1}(2^{\beta}-1)^{i}[T_q^a(\rho_{P|P_{j_i}})]^{\beta},\label{ftqa1}
	\end{align}
	Next we use the mathematical induction to prove the inequality (\ref{ftqa1}). When $m=2,$ similar to the proof of the Theorem \ref{tqa5}, we see that the theorem is correct. When $m\ge 2$, due to the condition (\ref{ftqa0}), we have
	\begin{align}
	0\le \frac{\sum_{k=1}^{m-1} T_q^a(\rho_{P|P_{j_k}})}{ T_q^a(\rho_{P|P_{j_0}})}\le 1,
	\end{align}
	then we have
	\begin{align}
	&	[T_q^a(\rho_{P|P_{j_0}\cdots P_{j_{m-1}}})]^{\beta}\nonumber\\
	\le & [T_q^a(\rho_{P|P_{j_0}})]^{\beta}[1+(2^{\beta}-1)(\frac{\sum_{k=1}^{m-1} T_q^a(\rho_{P|P_{j_k}})}{ T_q^a(\rho_{P|P_{j_0}})})^{\beta}]
	\nonumber\\
	=& [T_q^a(\rho_{P|P_{j_0}})]^{\beta} +(2^{\beta}-1)(\sum_{k=1}^{m-1}T_q^a(\rho_{P|P_{j_k}}))^{\beta},\label{tqa6}
	\end{align}
	At last, due to the mathematical induction, we have,
	\begin{align}
	(\sum_{k=1}^{m-1}T_q^a(\rho_{P|P_{j_k}}))^{\beta}\le \sum_{k=1}^{m-1}(2^{\beta}-1)^{k-1}[T_q^a(\rho_{P|P_{j_k}})]^{\beta},\label{tqa7}
	\end{align}
	combing the inequality (\ref{tqa6}) and (\ref{tqa7}), we finish the proof.
\end{proof}
\section{Conclusion}
\indent In this article, we investigate the general monogamy inequalities for the GW states in terms of T$q$EE. First we present an analytical formula for the T$q$EE and T$q$EEoA of the reduced density matrix of the GW state in terms of any partitions, then we present a monogamy inequality in terms of the squared T$q$EE for the reduced density matrices of the GW states, we also present a polygamy inequality in terms of the T$q$EE for the reduced density matrices of the GW states. Then we present generalized monogamy relations in terms of T$q$EE for the GW states. At last, we present tighter polygamy relations for the GWV states in terms of T$q$EEoA when $q=2.$ These results are meaningful as the GW states are in arbitrary n-qudit systems. Due to the importance of the study on the higher dimensional multipartite entanglement systems, and there are few results that the monogamy relations are valid for higher dimensional systems, our results can provide provide a reference for future work on the study of multiparty quantum entanglement. 
\bibliographystyle{IEEEtran}
\bibliography{re}
\end{document}